\documentclass[12p]{article}

\usepackage{amssymb}
\usepackage{graphicx}
\usepackage[utf8x]{inputenc}
\usepackage{enumerate}
\usepackage{paralist}
\usepackage{amsfonts}
\usepackage{amsmath}
\usepackage{stmaryrd}
\usepackage{hyperref}
\usepackage{amsthm}
\usepackage[table]{xcolor}
\usepackage{multirow}
\usepackage{framed}
\usepackage{authblk}

\hypersetup{
  pdftitle = {Scattered packings of cycles},
  pdfauthor = {A.\ Atminas, M.\ Kaminski, and J.-F.\ Raymond},
  colorlinks = true
}

\newcommand{\remref}[1]{\hyperref[#1]{(R\ref*{#1})}}
\newcommand{\itemref}[1]{\hyperref[#1]{(\ref*{#1})}}
\newcommand{\Itemref}[1]{\hyperref[#1]{Item (\ref*{#1})}}

\newtheorem{theorem}{Theorem}
\newtheorem{lemma}{Lemma}
\newtheorem{corollary}{Corollary}
\newtheorem{proposition}{Proposition}

 \theoremstyle{remark}
 \newtheorem{remark}{Remark}

 \theoremstyle{definition}
 \newtheorem{definition}{Definition}

\usepackage[boxed,vlined]{algorithm2e}
\SetAlgoCaptionLayout{small}

\newcommand{\floor}[1]{\left\lfloor#1\right\rfloor}
\newcommand{\ceil}[1]{\left\lceil#1\right\rceil}

\newcommand{\N}{\mathbb{N}}

\newcommand{\intv}[2]{\left \llbracket #1, #2 \right \rrbracket}
\newcommand{\card}[1]{\left | #1 \right |} 

\DeclareMathOperator{\vertices}{V} 
\DeclareMathOperator{\edges}{E} 
\DeclareMathOperator{\neigh}{{N}} 
\DeclareMathOperator{\maxdeg}{{\Delta}} 
\newcommand{\induced}[2]{#1\!\left [#2 \right ]} 
\DeclareMathOperator{\girth}{\mathbf{girth}} 

\title{Scattered packings of cycles\footnote{This work was partially
    supported by the Warsaw Center of Mathematics and Computer
    Science (Aistis Atminas and Jean-Florent Raymond) and the (Polish) National Science Centre
    under grants PRELUDIUM 2013/11/N/ST6/02706 (Jean-Florent Raymond) and SONATA 2012/07/D/ST6/02432 (Marcin Kami\'nski). Emails:
    \texttt{a.atminas@warwick.ac.uk}, \texttt{mjk@mimuw.edu.pl}, and \texttt{jean-florent.raymond@mimuw.edu.pl}.}}

\date{}
\author[1]{Aistis Atminas}
\author[2]{Marcin Kami\'nski}
\author[2,3]{Jean-Florent Raymond}

\affil[1]{DIMAP and Mathematics Institute, University of Warwick, Coventry, UK.}
\affil[2]{Institute of Informatics, University of Warsaw, Poland.}
\affil[3]{LIRMM, University of Montpellier, France.}

\begin{document}
\SetAlgorithmName{Algorithm}{Algorithm}{List of algorithms} 

\maketitle
\begin{abstract}
\noindent We consider the problem \textsc{Scattered Cycles}
  which, given a graph $G$ and two positive integers $r$ and $\ell$, asks whether $G$
  contains a collection of $r$ cycles that are pairwise at distance
  at least $\ell$. This problem generalizes the problem \textsc{Disjoint
  Cycles} which corresponds to
the case~$\ell = 1$. We prove that when parameterized
by $r$, $\ell$, and the maximum degree $\Delta$, the problem
\textsc{Scattered Cycles} admits a kernel on
$24\ell^2\Delta^\ell r \log(8 \ell^2\Delta^\ell r)$ vertices. We also provide a
$(16\ell^2 \Delta^\ell)$-kernel for the case $r=2$ and a $(148 \Delta r
  \log r)$-kernel for the case~$\ell = 1$. Our proofs rely on two simple
reduction rules and a careful~analysis.
\smallskip

\end{abstract}

\section{Introduction}
\label{sec:intro}

We consider the problem of deciding if a graph contains a collection
of cycles that are pairwise far apart. More precisely, given a graph~$G$ and two positive integers~$r$ and~$\ell$, we have to decide
if there are at least~$r$ cycles in~$G$ such that the distance between
any two of them is at least~$\ell$. By \emph{distance} between two
subgraphs $H,H'$ of a graph $G$ we mean the minimum number of edges in
a path from a vertex of~$H$ to a vertex of~$H'$ in~$G$.
This problem, that we call \textsc{Scattered Cycles}, is a
generalization of the well-known \textsc{Disjoint Cycles} problem,
which corresponds to $\ell=1$. It is also related to the \textsc{Induced Minor} problem\footnote{Given two graphs $H$ (guest)
  and $G$ (host),
  the \textsc{Induced Minor} problem asks whether $H$ can be obtained from an induced subgraph of $G$ by
contracting edges.} when $\ell=2$ in the sense that a graph contains $r$ cycles which are pairwise
at distance at least 2 if and only if this graph contains $r \cdot
K_3$ as induced minor. Hence, any result about the computational complexity of
\textsc{Scattered Cycles} gives information on the complexity
of~\textsc{Induced Minor}.
Lastly, it can be seen as an extension of the problems
\textsc{Independent Set} and \textsc{Scattered Set}, which instead of
cycles, ask for vertices which are far apart.

It is worth noting that besides the connections to other problems
mentioned above, this problem has several features which make its
study interesting. The first one is that neither positive instances,
nor negative ones are minor-closed classes. Therefore the tools from Graph
Minors (in particular the Graph Minor Theorem~\cite{Robertson2004325})
do not directly provide complexity results for this problem.

Also, \textsc{Scattered Cycles} seems unlikely to
be expressible in terms of the usual containment relations on graphs.
As pointed out above, the cases $\ell=1$ and $\ell=2$ correspond to
checking if the graph contains $r \cdot K_3$ as minor or induced
minor, respectively. However for $\ell>2$, none of the common
containment relations conveys the restriction that cycles have to be
at distance at least~$\ell$. Again, the techniques related to the
minor relation cannot be applied immediately.

Lastly, the special case $\ell = 2$ and $r = 2$ corresponds to a
question of~\cite{2013arXiv1309.1960C} (also raised in~\cite{Golovach09,
  belmonte}) about the complexity of checking whether a graph contains
two mutually induced cycles (equivalently, two triangles as induced minor).

Our goal in this paper is to investigate the kernelizability of \textsc{Scattered
  Cycles} under various parameterizations.
\autoref{tab:param} summarizes known results and the ones that we
obtained on the parameterized complexity of the problem with respect
to various combinations of parameters, among the number of 
cycles, the minimum distance required between two cycles and the
maximum degree of the graph. A parameterized problem is said to be
paraNP-hard if it is NP-hard for some fixed value of the parameter.
Unless otherwise specified, we will use all along the paper $r$, $\ell$ and $\Delta$ to
denote, respectively, the number of cycles, the minimum distance allowed between two
cycles and the maximum degree of the input graph.
 As the problem is unlikely to have
a polynomial kernel when parameterized by any of these parameters
taken alone (cf.\ \autoref{tab:param}) , we
naturally explore its kernelizability with several parameters. The
first column of the table counts the number of parameters taken into
account in a given~row, and ``par'' in the second column indicates
that the corresponding value is taken as parameter.
\begin{table}[ht]
  \centering
  \begin{tabular}{|c|c|c|c|l|}
    \hline
    \#par&$r$&$\ell$&$\Delta$&Complexity\\
    \hline \hline
    0& -- & -- & -- & NP-hard (\autoref{c:hard})\\ \hline
    \multirow{4}{*}{1}&par & $=1$ & -- &
                                         \parbox{8cm}{
                                         \begin{itemize}
                                         \item FPT (minor checking);
                                         \item no polynomial kernel unless
                                           $\mathrm{NP} \subseteq
                                           \mathrm{coNP}/\mathrm{poly}$~\cite{Bod09}.
                                         \end{itemize}
                                         }\\

    \cline{2-5}
         &par & $>1$ & -- & W[1]-hard (\autoref{c:hard})\\ \cline{2-5}
         &-- & par & -- & paraNP-hard (\autoref{c:hard})\\ \cline{2-5}
         & -- & -- & par &paraNP-hard (\autoref{c:hard})\\
    \hline
    \multirow{6}{*}{2}& par & $=1$ & par &
                                            \parbox{8cm}{
                                            \begin{itemize}
                                            \item $O(\Delta r \log
                                              r)$-kernel~\cite{fominLNGS11hitt};
                                            \item $(148\Delta r\log
                                              r)$-kernel (\autoref{c:dc}).
                                            \end{itemize}
                                            }\\\cline{2-5}
                                          
         &par & $>1$ & par&  Open\\
    \cline{2-5}
         &=2 & par & par & $(16\ell^2 \Delta^\ell)$-kernel (\autoref{t:2cycles})\\ \cline{2-5}
         &-- & par & par & para-NP-hard (\autoref{c:hard})\\ \cline{2-5}
         &par & par & -- &W[1]-hard (\autoref{c:hard})\\
    \hline
    3&par & par & par & $(24\ell^2\Delta^\ell r \log(8\ell^2\Delta^\ell
                        r))$-kernel (\autoref{t:manycycles})\\
    \hline
  \end{tabular}
  \caption{Complexity of \textsc{Scattered Cycles} wrt.\ various parameterizations.}
  \label{tab:param}
\end{table}

Our results are the following.

\begin{theorem}\label{t:manycycles}
  The problem \textsc{Scattered Cycles} admits a kernel on
  $24\ell^2\Delta^\ell r \log(8 \ell^2\Delta^\ell r)$ vertices when
  parameterized by $\ell$, $r$, and $\Delta$.
  Moreover this kernel can be computed from an $n$-vertex graph
  in~$O(n\ell)$~steps.
\end{theorem}

As mentioned above, a trivial consequence of
\autoref{t:manycycles} is that the problem of checking if a graph
$G$ contains $r\cdot K_3$ as induced minor admits a $O(\Delta^2 r \log (\Delta^2 r))$-kernel when
parameterized by $r$ and~$\Delta$.

\begin{theorem}\label{t:2cycles}
  The problem \textsc{Scattered Cycles} restricted to $r=2$
  admits a kernel on $16 \ell^2 \Delta^\ell$ vertices when parameterized
  by $\ell$ and ~$\Delta$.
  Furthermore this kernel can be computed from an $n$-vertex graph
  in~$O(n \ell)$~steps.
\end{theorem}
\autoref{t:2cycles} gives a partial answer to a question
of~\cite{2013arXiv1309.1960C} about the complexity of checking if
graph $H$ contains $2\cdot K_3$ as induced~minor.

The problem known as \textsc{Disjoint Cycles} corresponds to \textsc{Scattered
  Cycles} for~$\ell=1$. The authors of~\cite{Bod09} proved that when parameterized
by the number $r$ of cycles only, this problem does not have a
polynomial kernel unless ${\rm NP} \subseteq {\rm coNP}/{\rm poly}$.

For every graph $G$, let us denote by $\Lambda(G)$ the least non-negative integer $t$
such that $G$ does not contain $K_{1,t}$ as induced subgraph. This parameter refines the one of maximum
degree in the sense that for every graph $G$ we have~$\Delta(G)+1\geq
\Lambda(G)$ (hence $\{G,\ \Lambda(G)\leq k+1\}\supseteq\{G,\ \Delta(G)\leq k\}$).
 An $O(\Lambda r \log r)$-kernel has been provided for the problem
 \textsc{Disjoint Cycles} parameterized by the number $r$ of cycles
 and~$\Lambda$ in~\cite[Corollary~2]{fominLNGS11hitt}.
Building upon the techniques used to prove \autoref{t:manycycles}
and ideas from the proof of the aforementioned result, we achieved a
bound of the same order of magnitude, with a simpler proof and
explicit (small) constants.

\begin{theorem}\label{t:disjcyc}
  The problem \textsc{Disjoint Cycles} admits a kernel on $148 \Lambda r
  \log r$ vertices when parameterized by $\Lambda$
  and $r$. This kernel can be computed from an
  $n$-vertex graph in $O(n^2)$~steps.
\end{theorem}

As $\Delta(G)+1 \geq \Lambda(G)$ for every graph $G$, we immediately
obtain the following~corollary.
\begin{corollary}\label{c:dc}
  The problem \textsc{Disjoint Cycles} admits a kernel on $148 {(\Delta+1)} r
  \log r$ vertices when parameterized by $\Delta$
  and $r$. This kernel can be computed from an
  $n$-vertex graph in $O(n^2)$~steps.
\end{corollary}

The techniques used to obtain the above result do not translate to cases
$\ell>1$ and therefore cannot be directly used to improve
\autoref{t:manycycles}.

\paragraph*{Organization of the paper}
We introduce the reduction rules of our kernelization algorithm in
\autoref{sec:redux}, where we also consider packings of two distant
cycles and prove~\autoref{t:2cycles}. This result is generalized to
any number of cycles (\autoref{t:manycycles}) in
\autoref{sec:morecycles}. We investigate the special case $\ell
=1$ in \autoref{sec:disjcyc}. Lastly, \autoref{sec:hard} contains
the proofs of the hardness results that appear in \autoref{tab:param}.

\paragraph*{Discussion}
Using two simple reduction rules, we obtained a polynomial kernel for
the problem \textsc{Scattered Cycles} parameterized by the number of
cycles $r$, the distance required between any two cycles $\ell$ and
the maximum degree of the input graph. It should be noted that taken apart these parameters
are unlikely to give polynomial kernels. Furthermore, our proof is
constructive and the constants are small. The reduction also led to a
simplification of the proof of~\cite[Corollary~2]{fominLNGS11hitt}
that the problem \textsc{Disjoint Cycles} has a $O(\Lambda r \log
r)$-kernel, with in addition the computation of the constants.

A natural question is whether our upper-bound can be improved.
As we presented in the proof \autoref{l:lowerbnd}, there are reduced
graphs of order $(r-1) \left (
  \frac{d}{2}\right )^{\ell-1}$ which do not contain $r$ cycles pairwise
at distance at least~$\ell$. This suggests that other techniques or
reduction rules must be used in order to obtain a kernel of size $o(r \left (
  \frac{d}{2}\right )^{\ell-1})$ for this~problem.

\section{Reduced graphs}
\label{sec:redux}

\paragraph*{Basic definitions} Let $G$ be a graph. The
\emph{length} of a path is the number of edges it contains. The
\emph{distance} between two vertices $u,v\in \vertices(G)$ is the
minimum number of edges in a path from $u$ to~$v$. The \emph{girth} of a graph~$G$, denoted by $\girth(G)$, is the minimum length of a cycle in~$G$.
Two subgraphs $H,H'$ of $G$ are said to be \emph{$\ell$-distant} if
they are at distance at least~$\ell$. 
The set of non-negative integers is denoted by $\N$, and for every two $i,j \in \N$, we use
$\intv{i}{j}$ as a shorthand for the interval~$\{t \in \N,\ i \leq t \leq j\}$.

Let $\ell \in \N$ be a positive integer. An \emph{$\ell$-packing of $r$ cycles} in $G$ is a
collection of $r$ pairwise $\ell$-distant cycles of $G.$ These
packings generalize vertex-disjoint packings and induced packings
which respectively correspond to the cases $\ell = 1$ and~$\ell = 2.$
This notion enables us to formally define the problem that we consider,
as follows.

\begin{framed}
\noindent \textsc{Scattered Cycles}
  \begin{description}
\item[Input:] a graph $G$ and two integers $\ell \geq 1$ and $r\geq
  2$;
\item[Question:] Does $G$ have an \emph{$\ell$-packing of $r$ cycles}?
\end{description}
\end{framed}

\paragraph*{Reduced graphs} A \emph{subdivision path} in a graph $G$
is a subgraph of $G$ that is a path, and whose internal vertices are of
degree two. Let us call an edge $e$ of a graph
\emph{$\ell$-redundant} if it does not belong to a triangle, and if it
is an edge of a subdivision path of length more than~$\ell$. A graph is said to be
\emph{$\ell$-reduced} if it contains no $\ell$-redundant edge, nor a
vertex of degree zero or one. In order to obtain an $\ell$-reduced graph from
any graph, we consider the two following reduction~rules:
\begin{enumerate}[(R1)]
\item If the graph $G$ has an $\ell$-redundant edge, contract it.\label{e:r1} \medskip
\item If the graph $G$ contains a vertex $v$ of degree 0 or 1, delete it.\label{e:r2}\medskip
\end{enumerate}

It is clear that after an application of~\remref{e:r2} neither the set of cycles nor the distances between them are changed.
It is also easy to see that the contraction of an edge operation in \remref{e:r1} establishes a natural 1-1 correspondence between 
the set of cycles in the graph before and the set of cycles after an application of~\remref{e:r1}.
Moreover, one can also check that this correspondence preserves the property of being $\ell$-distant, i.e.
the two cycles are $\ell$-distant after an application of \remref{e:r1} if and only if the corresponding cycles before the application
of~\remref{e:r1} were $\ell$-distant. Hence for any positive integer $r$, these reduction rules do
not change the property of containing an $\ell$-packing of
$r$ cycles and also note that every graph can be reduced by a finite number of applications
of~\remref{e:r1} and~\remref{e:r2}. \autoref{a:redalg} is a
linear-time implementation of the reduction. Let us describe~it.

Intuitively, the set $S$ can be seen as a set of marked vertices, while
the vertices of $V \setminus S$ are not marked. The algorithm starts
from a graph where all vertices are unmarked (i.e. $S = \emptyset$)
and considers unmarked vertices while some exist in the graph. The
vertices of degree 3 that we may encounter are marked; those of degree
zero are deleted (i.e. we apply \remref{e:r2}). If a vertex $u$ of
degree two is incident with an $\ell$-redundant edge, then we contract
this (that is rule \remref{e:r1}) and keep the obtained vertex unmarked. That way we make sure that further contractions will be applied if the resulting subdivision path still has length more than~$\ell$.
Lastly, we delete every vertex $v$ of degree one (rule \remref{e:r2}). Observe that when the neighbor $u$ of $v$ has degree less than four, the deletion of $v$ might create a long subdivided path, or a vertex of degree one. By unmarking $u$ we ensure that these cases will be considered in a later step.

\begin{algorithm}
\DontPrintSemicolon
\KwIn{a graph $G$ and an integer $\ell$}
\KwOut{an $\ell$-reduced graph}
$V:=V(G)$\;
$S:= \emptyset$\; 
\While{$S \neq V$}{
pick $v \in V \backslash S$\;
\lIf{$\deg(v) \geq 3$}{$S:=S \cup \{v\}$}
\uElseIf{$\deg(v) = 2$}{
\uIf{for some $u \in \neigh(v),$ $N(u) \cap N(v)= \emptyset$ and
  $\{u,v\}$ belongs to a  subdivision path of length more than~$\ell$}
{contract edge $\{u,v\}$ and keep the resulting vertex in $V \backslash S$\;}
\lElse{$S=S \cup \{v\}$}
}
\uElseIf{$\deg(v) = 1$}{
\uIf{the only $u \in \neigh(v)$ is such that $u \in S$ and $\deg(u) \leq 3$}
{delete $v$ and delete $u$ from $S$\;}
\lElse{$V=V \backslash \{v\}$}
}
\lElse{delete $v$ (which in this case is isolated)}
}
\caption{Reduction algorithm.}
\label{a:redalg}
\end{algorithm}

\begin{lemma}\label{l:polyred}
  \autoref{a:redalg} runs in $O(n\ell)$ time and outputs an $\ell$-reduced graph.
\end{lemma}
\begin{proof}
  Every step in the while loop is performed in constant time, except
  checking if an edge belong to a subdivision path of length more than $\ell,$ which takes time $O(
  \ell)$. It is easy to check that each iteration of the while loop decreases the quantity $2|V|-|S|$ by 1 or 2, 
  and since $2|V|-|S| \geq |V| \geq 0$,
  the algorithm will perform the loop at most $2 |V(G)|$ times. This
  means that the algorithm is linear. Also notice, that the set $S$ does
  not contain any vertex of degree 1 nor any redundant edges.
  Hence the resulting graph is $\ell$-reduced as required.
\end{proof}

\begin{remark}\label{r:dmax}
 Reducing a graph does not increase its maximum degree.
\end{remark}

We now prove a generic lemma that will be used in the subsequent proofs.
\begin{lemma}\label{l:generic}
 Let $G$ be an $\ell$-reduced graph, let $r\in \N$, let $C$ be a cycle
 of $G$ of length $g=\girth(G)$. Let $h = \floor{\frac{g}{4}}$. For $i \in \intv{0}{\ell},$ we denote
 by $N_i$ the set of vertices at distance $i$ from $C.$ We also set $R
 = \vertices(G) \setminus (\vertices(C) \cup \bigcup_{i=1}^{\ell-1}
 N_i)$. Then we have:
  \begin{enumerate}[(i)]
  \item if $G$ has no $\ell$-packing of $r$ cycles, then $G[R]$ has no
    $\ell$-packing of $r-1$ cycles;\label{e:g1}
  \item $\forall i \in \intv{1}{h-1}$, every vertex of $N_i$ has
    exactly one neighbor in $N_{i-1}$;\label{e:g2}
  \item $\forall i \in \intv{1}{h-1}$, $N_i$ is independent;\label{e:g3}
  \item $\forall i \in \intv{1}{h-2}$, $|N_{i+1}| \geq |N_{i}|$;\label{e:g4}
  \item $\forall i \in \intv{1}{h-1-\ell}$, $|N_{i+\ell}| \geq
    2|N_{i}|$;\label{e:g5}
  \item if $g \geq 6$ then
  $|\vertices(G)| \geq g \cdot 2^{\floor{\frac{g}{4\ell}} - 1}$.\label{e:g6}
  \end{enumerate}
\end{lemma}

\begin{proof}
  \Itemref{e:g1} follows from the fact that every $\ell$-packing of $r-1$ cycles in
  $G[R]$ is $\ell$-distant from $C$, whereas we assume that $G$ has no
  $\ell$-packing of $r$ cycles.

  \textit{Proof of \Itemref{e:g2}.} By definition of $N_i$, for every integer $i$ such
  that $1 \leq i < h-1$, every vertex $v \in N_i$ has a neighbour in
  $N_{i-1}$.  Let us show that $v$ has exactly one neighbour in
  $N_{i-1}$. For this, we suppose for contradiction that $v$ has two
  neighbours $u, w \in N_{i-1}$ with $u \neq w$.  Then, there are two
  distinct paths of length $i$ from $v$ to the cycle $C$. If these paths
  have a vertex in common, then walking from $v$ along the first path
  until we reach the first vertex belonging to the second path and
  taking the second path back to $v$ would form a cycle of length at
  most $2i \leq 2 (h-1) < g$, a contradiction. On the other hand, if the
  two paths are vertex disjoint, consider their endpoints, say $x$ and
  $y$ which belong to the cycle $C$. In this case, the two paths
  together with the shortest path in $C$ between $x$ and $y$, which has
  length at most $\floor{\frac{g}{2}}$, would create a cycle of length
  $2i+\floor{\frac{g}{2}} \leq 2(h-1)+\floor{\frac{g}{2}} <g$, a
  contradiction. Thus we have proved that every vertex in $N_i$ has
  exactly one neighbour in $N_{i-1}$.

  \textit{Proof of \Itemref{e:g3}.} The argument is very similar to the one used in the
  proof of \Itemref{e:g2}: if there is an edge $\{u,v\}$ for some $u,v
  \in N_i$, then the paths respectively connecting $u$ and $v$ to $C$
  either intersect, what yields a cycle of length at most $1 +
  2(h-1)<g$, or they are vertex-disjoint, in which case we can build
  as above a cycle of length at most $1 + 2i+\floor{\frac{g}{2}} \leq
  1+ 2(h-1)+\floor{\frac{g}{2}} <g$.

  \textit{Proof of \Itemref{e:g4}.} As $G$ is $\ell$-reduced, every vertex
  of
  $N_i$ has degree at least two. Together with \Itemref{e:g2} and
  \Itemref{e:g3}, this implies that every vertex of $N_i$ has a neighbor
  in $N_{i+1}$. According to \Itemref{e:g2}, every distinct vertices
  $u,v \in N_i$, have disjoint neighborhoods in $N_{i+1}$. Hence
  $|N_{i+1}|\geq |N_i|$.

  \textit{Proof of \Itemref{e:g5}.}
Let us now show that the cardinality of the $N_i$'s is increasing as
follows: for every $i \in \intv{1}{h-\ell-1},$ $|N_{i + \ell}| \geq
2|N_i|$.
For every $v \in N_i$, let $S_v \subseteq N_{i + \ell}$ be the subset
of vertices at distance $\ell$ from $v$ in $\induced{G}{\bigcup_{p=i}^{i+\ell} N_p}$. By the structural description above if follows that every two distinct $u,v
\in N_i$ yields two disjoint non-empty sets $S_u$ and $S_v$. Also, by definition
of the $N_i$'s, every vertex of $N_{i+\ell}$ belongs to $S_u$, for some $u \in
N_i$. Therefore $\{S_u,\ u \in N_i\}$ is a partition of $N_{i + \ell}$
into $\card{N_i}$ subsets. If $\card{N_{i + \ell}} < 2 \card{N_i}$,
then there is a vertex $u \in N_i$ such that $S_u$ contains only one
vertex, that we call $v$. Let $w$ be the (unique) neighbor of $u$ in
$N_{i - 1}$. Then every interior vertex of the unique path linking $w$
to $v$ has degree two, and this path has length $\ell + 1$. This
contradicts the fact that $G$ is reduced, and thus~$|N_{i + \ell}| \geq
2|N_i|$.

  \textit{Proof of \Itemref{e:g6}.}
As a consequence of \Itemref{e:g5}, for every $i \in
\intv{1}{h-1},$ we have $|N_{i}| \geq
|N_{1}|\cdot 2^{\floor{(i - 1)/\ell}}.$ Since every subdivision path of $C$
has length at most $\ell$ (as $G$ is $\ell$-reduced), $|N_1| \geq
\ceil{\frac{g}{\ell}},$ hence for every $i \in \intv{1}{h-1},$
$|N_{i}| \geq \ceil{\frac{g}{\ell}} \cdot 2^{\floor{(i-1)/\ell}} \geq
\frac{g}{\ell} \cdot 
2^{\floor{(i-1)/\ell}}.$ We are now able to give a lower bound on the number of
vertices of $G$:
\begin{align*}
  |\vertices(G)| &\geq \sum_{i = 0}^{h - 1} |N_i|
  \geq g + \sum_{i = 1}^{h-1} \frac{g}{\ell}  \cdot 2^{\floor{(i-1)/\ell}}\\
  &\geq g + g \sum_{i = 0}^{\floor{\frac{h-2}{\ell}}} 2^i
  = g \cdot2^{\floor{\frac{h-2}{\ell}}+1}
  \geq g \cdot 2^{\floor{\frac{g}{4\ell}}-1}.
\end{align*}
\end{proof}

\begin{lemma}\label{l:tree-size}
  If $T$ is a tree with $s \geq 2$ leaves and no $\ell$-redundant edge,
  then $|\vertices(T)| \leq 2\ell s - 3 \ell +1$.
\end{lemma}

\begin{proof}
  Let $T$ be a tree as in the statement of the lemma and let $T'$ be the
  tree obtained from $T$ by dissolving every vertex of degree 2 (contracting an edge containing a vertex of degree 2 until no vertices of degree 2 are left). We
  denote by $s\geq 2$ the number of leaves in $T'$ (which remains the same
  as in~$T$) and by $t$ the number of internal vertices of $T'$.
  Since $T'$ is a tree, we have:
  \begin{align*}
    |\edges(T')| &= |\vertices(T')| - 1 = s + t - 1.
  \end{align*}
   This together with handshaking lemma and observation that every internal vertex has degree at least 3 imply:
\begin{align*}
    2(s + t -1) &= 2 |\edges(T')| = \sum_{v \in \vertices(T')}
    \deg(v) \geq s + 3t.
\end{align*}
Subtracting $2t$ from both sides we get the following bound on the number of vertices of~$T'$:
\begin{align*}
    2s -2 &\geq s + t = |\vertices(T')|.
  \end{align*}
  Now, observe that since $T$ does not contain an $\ell$-redundant edge, it has
  at most $(\ell - 1)|\edges(T')|$ vertices of degree two, and
  hence
  \begin{align*}
    |\vertices(T)| &\leq |\vertices(T')| + (\ell - 1)|\edges(T')|
                    = \ell |\vertices(T')| - \ell + 1
                    \leq 2\ell s - 3 \ell +1.
  \end{align*}
 
\end{proof}

\begin{corollary}\label{c:forest-size}
  If $F$ is a forest with $s \geq 2$ leaves or isolated vertices and
  without $\ell$-redundant edges, then $|\vertices(F)| \leq 2 \ell s -
  3\ell +1.$
\end{corollary}

\begin{proof}
  First observe that if $F$ has no connected component of order at
  least~3, then we have $|V(F)|=s \leq 2ls-3l+1$. The latter inequality holds for all $s \geq
  2$ and $l \geq 1$ and one can verify it by observing that it is equivalent to $0 \leq 2(l-1)(s-2)+(l-1)+(s-2)$. 
 On the other hand, if $F$ has a connected component of order at
  least~3, we can add edges between internal vertices of different
  connected components of order at least~3 in order to obtain a forest
  with the same vertex set and the same number of isolated vertices
  and leaves and containing exactly one tree $T$ on at least 3
  vertices. If $s_1$ is the number of leaves in $T$, then by
  \autoref{l:tree-size} we have $|V(T)| \leq 2\ell s_1 - 3\ell +
  1$. The rest of the forest (consisting of components of order 1 and
  2) contains $s - s_1$ vertices. Hence $|V(F)| \leq 2\ell s_1 - 3\ell +
  1 + s - s_1 \leq 2\ell s - 3\ell + 1$.
\end{proof}

\begin{lemma}\label{l:v-ub1}
  If~$G$ is an $\ell$-reduced graph not containing two
  $\ell$-distant cycles, then $|\vertices(G)| < 2\ell \girth(G)
  \maxdeg(G)^{\ell}.$
\end{lemma}

\begin{proof}
  Observe that as $G$ is $\ell$-reduced, it contains two $\ell$-distant
  cycles as soon as it has more than one connected
  components. Therefore we shall now assume that $G$ is connected.
  Let $C$ be a cycle in $G$ of length $g=\girth(G)$. We define $N_i$ for every $i \in
  \intv{0}{\ell}$ and $R$ as in \autoref{l:generic} and we set $\Delta = \maxdeg(G)$. According to \hyperref[e:g1]{Item~(\ref*{e:g1}) of \autoref*{l:generic}}, $R$ induces a forest in $G$.
  Also notice that every
  leaf or isolated vertex of $\induced{G}{R}$ belongs to $N_{\ell},$
  otherwise it would have degree at most one in $G,$ which would contradict the fact
  that $G$ is $\ell$-reduced. Besides, if $\induced{G}{R}$ has a subdivision path
  of length more than $\ell$, at least one of its internal vertices must belong to
  $N_\ell$ (and have neighbors in $N_{\ell-1}$), otherwise it would
  contradict the fact that $G$ is $\ell$-reduced. Let us consider the graph $R^+$ constructed from $\induced{G}{R}$ by adding a neighbor of
  degree one to each vertex of $N_\ell$ which has degree two in $G[R]$. Now $R^+$ is a
  forest which has at most $|N_\ell|$ leaves or isolated vertices and has
  no $\ell$-redundant edge: by \autoref{c:forest-size} we
  have~$|\vertices(R^+)| \leq 2\ell |N_\ell| - 3\ell + 1$.

  The cycle $C$ has $g$ vertices each of degree at most $\Delta$
  and with two neighbors in $C,$ therefore $|N_1| \leq g (\Delta -
  2)$ and by a similar argument we obtain~$|N_i| \leq g(\Delta -
  2)(\Delta - 1)^{i-1}$ for every $i\in \intv{2}{\ell}$. We are now able
  to give an upper-bound on the order of $G$:
  \begin{align*}
  |\vertices(G)| &= |\vertices(C)| + \sum_{i = 1}^{\ell - 1}|N_i| + |R| \\
  & \leq |\vertices(C)| + g((\Delta-1)^{\ell - 1} - 1) + |\vertices(R^+)|\\
  & \leq g(\Delta-1)^{\ell - 1} + 2\ell g(\Delta -
  2)(\Delta - 1)^{\ell-1} -3\ell +1\\
  & \leq 2 \ell g \Delta^\ell - g (\Delta-1)^{\ell - 1} -2\ell\\
  & < 2 \ell g \Delta^\ell.
  \end{align*}
\end{proof}

Now we show that reduced graphs without two $\ell$-distant cycles must have small girth.
\begin{lemma}\label{l:small-g}
  If~$G$ is an $\ell$-reduced graph not containing two
  $\ell$-distant cycles, then $\girth(G) \leq 8 \ell - 4$.
\end{lemma}

\begin{proof}
  Let us assume by contradiction that $G$ has girth $g > 8 \ell - 4$. We use the
  same notation for $C, R$ and $N_i$ (for every $i \in \intv{0}{\ell}$) as in
  \autoref{l:generic}.
As in \autoref{l:v-ub1}, $R$ induces a forest, all the leaves and isolated vertices of which lie in~$N_\ell$. But since $N_\ell$ is
independent and each $v \in N_\ell$ has exactly one neighbor in
$N_{\ell-1}$ (by \Itemref{e:g2} and
  \Itemref{e:g3} of \autoref{l:generic}, as $\ell<\floor{\frac{g}{4}}$) we
deduce that $G[R]$ contains only components of order at
least~3. Moreover, if we pick any leaf $v\in N_\ell$, there is in
$\induced{G}{R}$ a vertex of degree at least 3 which is at distance at most
$\ell-1$ from $v$, otherwise we would either find a vertex $u \in R
\setminus N_\ell$ of degree one in $G$, or an $\ell$-redundant edge,
thus contradicting the fact that $G$ is~reduced.

Having learned the structure of the graph, we are ready 
to derive a contradiction on the value of the girth as follows. 
Pick an arbitrary component in $G[R]$ and a path 
$P$ of maximal length in it. Let $v$ and $t$ be the two endpoints of the path. Let $u$ be a vertex of degree at least three of minimal distance from~$v$. Observe that such a vertex is unique and belongs to $P$. According to the previous paragraph, $u$ is at distance a most $\ell-1$ from the leaf $v$. Let $v'$ be a vertex
of maximal distance reachable from $u$ in $\induced{G}{R \setminus (P \setminus
  \{u\})}$, i.e. let $v'$ be in the same connected component of $\induced{G}{R \setminus (P \setminus
  \{u\})}$ as $u$ with the longest possible distance from $u$. Observe that $v'$ is a leaf and that that the distance
between $t$ and $v'$ in $G[R]$ is at most $|P|$ (by maximality of
$P$). Therefore, $v$ and $v'$ are at distance at most $2(\ell - 1)$ in
$\induced{G}{R}$. Let $Q$ be the unique path linking $v$ to $v'$ in
the forest $\induced{G}{R}$.
Let $P_v$ (resp.\ $P_{v'}$) be a shortest path from $v$ (resp.\ $v'$)
to~$C$ and let $w$ (resp.\ $w'$) be the endpoint of $P_v$ (resp.\
$P_{v'}$) in~$C$. Note that since $v,v' \in N_\ell$, both of these paths have length~$\ell$.
Let $Q'$ be the
shortest subpath of $C$ linking $w$ to $w'$ and observe that $Q'$ has
length at most $\floor{\frac{g}{2}}$.
The subgraph $\induced{G}{\vertices(Q) \cup \vertices(P_v) \cup
  \vertices(P_{v'}) \cup \vertices(Q')}$ clearly contains a
cycle. This subgraph has at most $2(\ell - 1) + 2\ell +
\frac{g}{2}$ edges hence $g\leq 8\ell -4$, a contradiction.
\end{proof}

 Combining
  \autoref{l:v-ub1} and~\autoref{l:small-g}, we obtain the following result.
\begin{corollary}\label{c:small-graph}
 If $G$ is an $\ell$-reduced graph not containing two $\ell$-distant cycles,
 then $|\vertices(G)| < 2\ell(8\ell - 4) \maxdeg(G)^\ell.$
\end{corollary}

We are now ready to prove \autoref{t:2cycles}.

\begin{proof}[Proof of \autoref{t:2cycles}]
  Consider the following procedure.
  Given a graph $G,$ and two integers $r$ and $\ell$, we apply
  \autoref{a:redalg} and obtain a graph $G'.$ If
  $\card{\vertices(G')} \geq 2\ell(8\ell - 4) \maxdeg(G)^\ell,$ then
  we output the graph $K_3 + K_3,$ otherwise we output $G'.$
  The call to the reduction algorithm runs in $O(\card{\vertices(G)}\ell)$-time, as
  explained in~\autoref{l:polyred}. Moreover, observe that either the
  procedure outputs the $\ell$-reduced input graph, or $K_3 + K_3$ in
  which case, the input graph is known to contain two $\ell$-distant
  cycles, by \autoref{c:small-graph}. According to \autoref{r:dmax}, the
  maximum degree of an $\ell$-reduced graph is never more than the one of the
  original graph. Therefore the output instance is equivalent to the
  input with regard to the considered problem. At last,
  the output graph has order upper-bounded by~$2\ell(8\ell - 4)
  \maxdeg(G)^\ell$. This proves the existence of a $(16\ell^2 \Delta^\ell)$-kernel for this~problem.
\end{proof}

\section{Dealing with more cycles}
\label{sec:morecycles}

In this part, we focus on the structure of graphs not containing an
$\ell$-packing of $r$ cycles, for some fixed positive integers~$r$ and~$\ell$. Using the ideas of
the above section, we show that the problem \textsc{Scattered
  Cycles} parameterized by $\Delta$, $r$, and $\ell$ admits a $O(\ell^2\Delta^\ell r
\log(\ell^2\Delta^\ell r))$-kernel.
\begin{definition}
For positive integers $\ell\geq 1, r \geq 1, d$ we denote by
$h^\ell_r(d)$ the least integer such that every $\ell$-reduced
graph $G$ of degree at most $d$ and with more than $h^\ell_r(d)$ vertices
has an $\ell$-packing of $r$ cycles. When such a number does not exist,
we set~$h^\ell_r(d) = \infty.$
\end{definition}

We showed in the previous section that $h_2^\ell(d) \leq 2\ell(8\ell - 4)
  d^\ell$ and it is easy to see that~$h_1^\ell(d)=1$. In this section we will show that for every $\ell\geq 2$,
  $r \geq 2$, $d \geq 1$ we have~$h^\ell_r(d) \leq
  24\ell^2 d^\ell r\log(8\ell^2 d ^\ell r)$. Let us first give
  a lower bound on~$h^\ell_r(d)$.
\begin{lemma}\label{l:lowerbnd} 
  For every $\ell \geq 1,\ r \geq 2, d \geq 2$, $h^\ell_r(d) \geq (r-1) \floor{\frac{d}{2}}^{\ell-1}$.
\end{lemma}

\begin{proof}
  We start with $r = 2$. If $d \in \{2, 3\}$, then set $G=C_3$, a cycle on 3 vertices. For any $\ell \in \mathbb{N}$, $G$ is $\ell$-reduced by definition and clearly does not contain two $\ell$-distant cycles. 
By definition of $h_2^{\ell}(d)$, it follows that $G$ must have at most $h_2^{\ell}(d)$ vertices and we obtain that
$h^\ell_2(d) \geq |V(G)| = 3 > 1=\floor{\frac{d}{2}}^{\ell-1}$ holds for $d \in \{2,3\}$ and any $\ell \in \mathbb{N}$. 
Similarly, taking $G=C_3$, we can settle the lemma for $\ell=1$ and $d \geq 2$. Suppose now $d \geq 4$, $\ell \geq 2$ and  let $G$ be an undirected 
de Bruijn graph of type
  $(\floor{\frac{d}{2}}, \ell - 1)$, which is a regular graph of degree $2\floor{\frac{d}{2}}$,
 diameter $\ell - 1$ and order~$\floor{\frac{d}{2}}^{\ell-1}$ (cf.\ \cite[Section 2.3.1]{Miller13mooregraphs}
  for definition and properties). As the diameter of $G$ is $\ell -1$, $G$
  does not contain two $\ell$-distant cycles and since each of its vertices has degree $2\floor{\frac{d}{2}} \geq 4$, $G$ must be $\ell$-reduced. As before, we conclude that the graph
  $G$ must have at most $h^\ell_2(d)$
  vertices which establishes $h^\ell_2(d) \geq |V(G)| =
  \floor{\frac{d}{2}}^{\ell-1}$.
  Let us now consider the case $r>2.$ and let
  $G_r$ be the disjoint union of $r-1$ copies of the graph $G$.
  According to the remarks above, $G_r$ is $\ell$-reduced and does not contain an
  $\ell$-packing of $r$ cycles. Hence, $h^\ell_r(d) \geq \card{\vertices(G_r)} = (r-1) \floor{\frac{d}{2}}^{\ell-1}$.
\end{proof}

For every $r,d,\ell$ positive integers, let $f^\ell_r(d)
= 24 \ell^2 d^\ell r \log(8 \ell^2 d^\ell r)$. The following lemma states
that every $\ell$-reduced graph with degree at most $d$ either
contains an $\ell$-packing of $r$ cycles, or has size at most~$f^\ell_r(d)$.
\begin{lemma}\label{l:smallrcyc}
For every positive integers $r \geq 2$ and $d,$ we have $h^\ell_r(d) \leq f^\ell_r(d)$.
\end{lemma}

\begin{proof}
Let $r \geq 2$, $d, l \in \mathbb{N}$ be arbitrary positive integers
and consider a graph $G$ which is $\ell$-reduced, with maximum degree
at most $d$ and not containing an $\ell$-packing of $r$ cycles.
We use the same notation for $C, R$ and $N_i$ (for every $i \in
\intv{0}{\ell}$) as in \autoref{l:generic}.
Recall that $\induced{G}{R}$ does not contain an $\ell$-packing of $r-1$
cycles (\hyperref[e:g1]{Item~(\ref*{e:g1})} of \autoref*{l:generic}).

Notice that $R \setminus N_\ell$ does not contain a vertex of degree
less than two nor an edge that is $\ell$-redundant in $G[R]$. In what
follows, we will reduce the graph $\induced{G}{R}$ to the graph
$R^+$. Since $R^+$ is $\ell$-reduced graph without an
$\ell$-packing of $r-1$ cycles, it has bounded order, by
induction. From this we will conclude the bound on $|R|$ and hence the
bound on~$|\vertices(G)|$.
Now, we need to count the number of vertices lost in reduction
procedure. To make the calculation easier, we consider the slightly
modified reduction routine \autoref{a:redalg2}.

\begin{algorithm}
\DontPrintSemicolon
\KwIn{a graph $G$ and the sets $N_\ell$ and $R$}
\KwOut{the graph $R^+$}
$N:=N_\ell$\;

\While{$N$ contains a vertex $v$ of degree one}{
  Let $P$ be the longest subdivision path in $G[R]$ whose length is at most $\ell+1$.
  Contract all the edges of $P$ and keep the resulting vertex in N. }
\While{$N$ contains a vertex $v$ of degree two}{
Let $u_1, u_2$ be the two neighbors of $v$ and let $P$ be the maximal
subdivision path going through $v$, of length at most $2\ell+1$. \; 
\lIf{$\card{\edges(P)} \leq \ell$}{remove $v$ from $N$.}
\lIf{$\ell < \card{\edges(P)} \leq 2 \ell$}{contract a subpath of $P$ with $\card{P} - \ell$ edges (and including vertex $v$) into a single vertex and keep the resulting vertex in $R\backslash N$.}
\lIf{$\card{\edges(P)} = 2\ell+1$}{let $P'$ be a subpath of $P$ of length $l+1$ starting at vertex $v$ (and going through either $u_1$ or $u_2$). Contract all edges of the path $P'$ into a single vertex and keep it in $N$.}

}
\While{$N$ contains an isolated vertex $v$}{
delete $v$\;
}

\caption{Reduction of $\induced{G}{R}$.}
\label{a:redalg2}
\end{algorithm}

Let us briefly describe this routine, which works on the graph $G[R]$.
We consider a set $N \subseteq R$ with the property that every time we can apply the rule \remref{e:r1} or \remref{e:r2} to $G[R]$, there is a vertex of $N$ where the rule can be applied. We will make sure that this property is an invariant of the algorithm. Hence, the graph will be reduced when none of the rules will be applicable to a vertex of $N$. As $G$ is $\ell$-reduced, $N_\ell$ satisfies the above property as it contains all vertices of $R$ with a neighbor in $V(G) \setminus R$. Therefore we start the algorithm with $N=N_\ell$.
If $N$ contains a vertex $v$ of degree one, \remref{e:r2} allows us to delete it. To make the counting easier, we also delete vertices along the maximum subdivision path of length at most $\ell$ starting from $v$, which is also allowed by \remref{e:r2}. Similarly to what we do in \autoref{a:redalg}, we need to add the neighbor of the last deleted vertex to $N$ because a reduction might be applicable to it in a later step.
If $N$ has a vertex of degree 2 that is incident with an
$\ell$-redundant edge, we can apply \remref{e:r1} to contract this
edge. Again, we contract more edges to make the calculations easier
but we make sure that each edge we contract satisfies the requirements
of \remref{e:r1}. We also add a vertex incident to the lastly
contracted edge to~$N$, for the same reason as previously. Besides, the isolated vertices belonging to
$N$ can be deleted from the graph.
After completing these steps, $N$ contains only vertices where none of
our reductions rules can be applied: the graph is reduced. Let us now
count vertices lost during the reduction.
            
Let $d_1$ be the initial number of vertices of degree one in $G[R]$. As
there are no vertices of degree one in $R \backslash N_\ell$ we have $d_1
\leq |N_\ell|$.  
It is not hard to see that after each step of the first while loop of
\autoref{a:redalg2}, the quantity $|N|+d_1$ decreases by at
least one. Notice also, that after each step of the second or third
while loop the quantity $|N|$ decreases by at least one. To see this for the case $\card{E(P)}=2\ell+1$ in the second while loop, it is enough to note that the path $P'$ must have at least 2 vertices in $N$ as otherwise $P'$ is a subdivision path of length more than $\ell$ in $G$, which is not possible as $G$ is $\ell$-reduced.  Hence, all
in all, at most $|N|+d_1$ steps are performed in the reduction
algorithm. Now, notice that each step reduces the number of vertices
in $G[R]$ by at most $\ell$. Hence, at the end of the algorithm we will
have an $\ell$-reduced graph $R^+$ such that $|R|-|\vertices(R^+)| \leq \ell(|N_\ell|+d)
\leq 2\ell |N_\ell|$.

The graph $R^+$ is $\ell$-reduced and does not contain an
$\ell$-packing of $r-1$ cycles: by definition of $h^\ell_{r-1}$ we have $|R^+| \leq h^\ell_{r-1} (d)$. Putting
these bounds together, we obtain an~inequality:
\begin{align}
 |\vertices(G)| &= |C|+\sum_{1 \leq i < \ell} |N_i| +|R| \nonumber\\
&\leq |C| + \sum_{1 \leq i < \ell} |N_i| + 2\ell |N_\ell| + |R^+| \nonumber\\ 
&\leq g + \sum_{1 \leq i < \ell} g(d-2)(d-1)^{i-1} + 2\ell g(d-2)(d-1)^{\ell-1} + h^\ell_{r-1}(d) \nonumber\\ 
&\leq 2\ell g d^\ell + h^\ell_{r-1}(d) \label{eq:star}
\end{align} 
Now observe that when $g > 4\ell \log(2 \ell d^\ell +
h_{r-1}^\ell(d)) + 8\ell$, by \hyperref[e:g5]{Item (\ref*{e:g5}) of \autoref*{l:generic}} we get:
\begin{align*}
  |\vertices(G)|   &\geq g 2^{\frac{g}{4\ell}-2} \\
           &> g 2^{\log(2 \ell d^\ell +h_{r-1}^\ell(d))} \\
           &\geq g(2 \ell d^\ell +h_{r-1}^\ell(d))\\
           &> |\vertices(G)|&(\text{using}\ \itemref{eq:star})
\end{align*}
This contradiction leads to the conclusion that $g \leq 4\ell \log(2 \ell d^\ell +
h_{r-1}^\ell(d)) + 8\ell$ and putting this bound on the girth of $G$ into~\itemref{eq:star} we~get:
\begin{align*}
|\vertices(G)| \leq 8 \ell^2 d^\ell \log(8\ell d^\ell +
  4h^\ell_{r-1}(d)) + h^\ell_{r-1}(d).
\end{align*}
As this holds for every $\ell$-reduced graph without $r$ $\ell$-distant cycles with degree bounded by $d$ we obtain:
\begin{align}
h_r(d) & \leq 8 \ell^2 d^\ell \log(8\ell d^\ell +
  4h^\ell_{r-1}(d)) + h^\ell_{r-1}(d). \label{eq:starstar}
\end{align} 
To finish the proof, we will check by induction on $r$ that $h_r(d)$
is at most~$f^\ell_r(d)$.
It is true for $r=2$ by \autoref{c:small-graph}. Suppose $r>2$
and $f^\ell_{r-1}(d) \geq h^\ell_{r-1}(d)$, and let $D=8\ell^2 d^\ell$ for
convenience. Then we have the following. 
\begin{align*}
f^\ell_r(d) - h^\ell_{r-1}(d)& \geq f^\ell_r(d) - f^\ell_{r-1}(d)  &\text{(induction hypothesis)}\\
		      & = 3D r\log (Dr) - 3D (r-1)\log(D(r-1)) \\
                           & \geq D \log((Dr)^3) \\
                           & \geq D \log(4\ell^2d^\ell + (8\ell^2d^\ell r) (4\ell^2d^\ell r)(8\ell^2d^\ell r))  \\
                           & \geq D \log(4 \ell^2 d^\ell + 96\ell^2
                             d^\ell r \log(8 \ell^2 d^\ell r) )&\text{(term by term, $r \geq 3$)} \\
		     & = 8\ell^2 d^\ell \log(8 \ell^2 d^\ell + 4f^\ell_r(d)) \\
                           & \geq 8\ell^2 d^\ell \log(8 \ell^2 d^\ell + 4h^\ell_r(d)) &\text{(induction hypothesis)}
\end{align*}

Together with~\itemref{eq:starstar} this implies: $f^\ell_r(d) \geq h^\ell_{r-1}(d) + 8\ell^2 d^\ell \log(8 \ell^2 d^\ell + 4h^\ell_r(d))
\geq h^\ell_r(d)$. Hence we are done.
\end{proof}

We are now able to prove \autoref{t:manycycles}.

\begin{proof}[Proof of \autoref{t:manycycles}]
  Given a graph $G$ and two integers $r$ and $\ell$, we apply
  \autoref{a:redalg2} to obtain in $O(\card{\vertices(G)}\ell)$
  steps an $\ell$-reduced graph $G'$ with $\Delta(G) = \Delta(G')$, as
  explained in \autoref{l:polyred} and \autoref{r:dmax}.

  If $\card{\vertices(G')} \geq f^\ell_r(\maxdeg(G))$, then by the
  virtue of \autoref{l:smallrcyc} the graph $G'$
  contains an $\ell$-packing of $r$ cycles, and then so do $G$. In this
  case we output the equivalent instance $(r \cdot K_3, \ell, r)$ and
  otherwise we output $(G', \ell, r)$.
  Observe that order of $G'$ is bounded by $f^\ell_r(\maxdeg(G)),$ a
  function of its maximum degree, $\ell$, and $r$ which are the
  parameters of this instance. This proves the existence
  of a kernel on $f^\ell_r(\Delta) = 24 \ell^2 \Delta^\ell r \log(8 \ell^2
  \Delta^\ell r)$ vertices for the problem \textsc{Scattered Cycles}
  parameterized by~$\Delta$, $\ell$, and~$r.$
\end{proof}

\section{The case of Disjoint Cycles}
\label{sec:disjcyc}

This section is devoted to the proof of \autoref{t:disjcyc},
which is similar in flavour with the proof
of~\cite[Corollary~2]{fominLNGS11hitt}. A \emph{feedback-vertex-set}
(\emph{fvs} for short) of a graph $G$ is a set of vertices meeting all
the cycles of $G$.
The proofs we will present here rely on the following results.

\begin{proposition}[Erdős-Pósa Theorem~\cite{EP62}]\label{p:ep}
  Let $f \colon \N_{\geq 1} \to \N$ be defined by $f(1) = 3$ and for every $k
  >1$, $f(k) = 4k(\log k + \log\log k + 4) + k - 1$.

  For every integer $k>1$, every graph contains either $k$ disjoint
  cycles, or a fvs of at most $f(k)$~vertices.
\end{proposition}

\begin{remark}\label{r:bound}
  For every $k\geq 3$, we have \[f(k) \leq c \cdot k \log k,\quad \text{where}\ c = \frac{17+4 \log 3 + 4 \log \log 3}{ \log 3} <16.4.\]
\end{remark}

\begin{proposition}[\cite{BafnaBF99}]\label{p:fvs}
  There is an algorithm that given an $n$-vertex graph computes in
  $O(n^2)$-time a 2-approximation of a minimum fvs.
\end{proposition}

\begin{proof}[Proof of \autoref{t:disjcyc}]
  Let us describe the steps of a kernelization algorithm for
  \textsc{Disjoint Cycles}.
  We are given a graph $G$ and an integer $r$.
  We assume that $r>1$, and $\Lambda(G) >1$, otherwise the problem is
  trivially solvable in polynomial time. If $r=2$, then we use the
  algorithm of \autoref{t:2cycles}.
  Let $G'$ be the graph
  obtained by the application on $G$ of the reduction routine
  \autoref{a:redalg}, for $\ell=1$.
  Using the algorithm of \autoref{p:fvs}, we compute a
  2-approximation of a minimum fvs $X$ of $G'$.
  If $|X| > 2f(r)$, by \autoref{p:ep} the graph $G'$ contains $r$
  disjoint cycles (and so do $G$): we return the equivalent positive instance
  $(r\cdot K_3, r)$. Otherwise, we return~$(G', r)$.

  Let us now bound the order of $G'$ in the latter case.
  Let $N = \neigh_{G'}(X) \setminus X$ (the neighbors of $X$
  outside $X$) and let $R$ be the graph obtained from the forest $G'
  \setminus X$ by adding a neighbor of degree one to every vertex of
  $N$ that has degree two in $G' \setminus X$. Observe that $R$ is a
  $1$-reduced forest with $N$ leaves: by \autoref{c:forest-size}
  we get $|\vertices(R)| \leq  2 \card{N} - 2$.
  Let $\Lambda = \Lambda(G)$.
  As $X$ is a fvs, for every vertex $u \in X$ the induced subgraph $G[\neigh(u)
  \setminus X]$ is a forest. It is well-known that any forest has an
  independent set on at least half of its vertices. Therefore, $|G[\neigh(u)
  \setminus X]| < 2\Lambda$, otherwise $G$ would contain an induced $K_{1,\Lambda}$.
  We can then deduce that $|N| < 2 \Lambda |X|$.
  We are now able to bound the order of $G'$, also using the fact that
  $R$ is a supergraph of $G' \setminus X$:
\begin{align*}
  \card{\vertices(G')} & = \card{X} + \card{\vertices(G' \setminus X)}\\ 
                       & \leq \card{X} + \card{\vertices(R)}\\
                       & < \card{X} + 4 \Lambda \card{X} - 2\\
                       & < 9 \Lambda f(r) & \text{(as $\Lambda > 1$)}\\
  \card{\vertices(G')} & < 148 \Lambda r \log r& \text{(using \autoref{r:bound})}
\end{align*}
\end{proof}

\section{Hardness}
\label{sec:hard}
This section contains the proofs of the hardness results claimed in
\autoref{tab:param}. Let us first define the problem \textsc{Independent Set}, as most of our proofs relies on its properties.

\begin{framed}
\noindent \textsc{Independent Set}
  \begin{description}
\item[Input:] a graph $G$ and an integer $r\geq 2$;
\item[Question:] Does $G$ have a collection of $r$ pairwise
  non-adjacent vertices?
\end{description}
\end{framed}

The known facts that we will use about independent set are the following.
\begin{proposition}\label{p:is}
  \textsc{Independent Set}~is
  \begin{inparaenum}[(i)]
  \item NP-hard, even when restricted to
    graphs of maximum degree~3~\cite{GareyJohnson}; and\label{e:isnp}
  \item W[1]-hard when parameterized by~$r$~\cite{ParameterizedComplexity}.\label{e:isw1}
  \end{inparaenum}
\end{proposition}

\begin{lemma}\label{l:red}
  For every instance $(G,r)$ of \textsc{Independent Set} and for every
  $\ell \in \N$, $\ell \geq 2$ we can construct in $O(\ell |\edges(G)|)$ steps
  an instance $(G', \ell,r)$ of \textsc{Scattered Cycles} with
  $|\vertices(G')| = O(\ell |\edges(G)|)$ and $\Delta(G') =
  \Delta(G)+2$ such that
  $(G,r)$ is a positive instance iff $(G', \ell, r)$ is a positive instance.
\end{lemma}

\begin{proof}
  Let $G'$ be the graph
  obtained from $G$ by subdividing every edge $\ell - 2$ times and for
  every vertex $v$ of the original graph adding the two vertices $v'$
  and $v''$ and the three edges $\{\{v,v'\},\{v',v''\}, \{v'',
  v\}\}$ (calling $C_v$ the obtained triangle).
  This construction requires to add $\ell$ new vertices for each edge of
  $G$ and a constant number of new vertices for each vertex of $G$, hence it can
  be performed in $O(\ell|\edges(G)|)$ steps. 
  Observe that $|\vertices(G')| = 3|\vertices(G)| + (\ell-2)
  |\edges(G)| = O(\ell|\edges(G)|)$ and $\Delta(G') = \Delta + 2$.
  Let us show that for every $r \in \N$, $G$ has an independent set of size $r$ iff $G'$ has an $\ell$-packing
of $r$ cycles.

Direction~``$\Rightarrow$''. Let $\{v_1, \dots, v_r\}$ be an independent
  set of size $r$ in $G$. Then $\{C_{v_1}, \dots,
  C_{v_r}\}$ is an $\ell$-packing of $r$ cycles. Indeed,
  by definition of $G'$ for every $i \in \intv{1}{r}$, the graph
  $C_{v_i}$ is a triangle. Besides, for every $i,j
  \in \intv{1}{r}$, $i \neq j$, the vertices $v_i$ and $v_j$ are at
  distance 2 in $G$, hence $C_{v_i}$ and $C_{v_j}$ are at distance at
  least~$2 \ell -4 \geq \ell$.

  Direction~``$\Leftarrow$''. Let $\mathcal{S} = \{S_1, \dots, S_r\}$ be an $\ell$-packing of
  $r$ cycles in $G'$. Observe that there is no cycle in $G'$ no vertex
  of which belongs to the original graph $G$. Therefore for every $i \in \intv{1}{r}$ the
  subgraph $S_i$ contains a vertex $v_i$ which belong to~$G$. Moreover,
  for every $i,j \in \intv{1}{r}$, $i \neq j$, the vertices $v_i$ and
  $v_j$ are at distance at least $\ell$ in $G'$ (as $\mathcal{S}$ is a
  $\ell$-packing), thus they are at distance at least
  2 if $\ell = 2$ and $\ceil{\frac{\ell}{\ell-2}} \geq 2$ otherwise in $G$. Consequently $\{v_1,
  \dots, v_r\}$ is an independent set of size $r$ in~$G$.
\end{proof}

\begin{corollary}\label{c:isred}
  For every $\ell\geq 2$, if there is an algorithm solving the problem
  \textsc{Scattered Cycles} in $f_\ell(r,\Delta,n)$ steps (where $n$ is the
  order of the input graph and $\Delta$ its maximum degree) for some
  function $f_\ell \colon \N^3 \to \N$, then there is an algorithm solving
  \textsc{Independent Set} in at most $(f_\ell(r,\Delta  + 2, n^{O(1)}) + n^{O(1)})$ steps.
\end{corollary}

\begin{corollary}\label{c:hard}
\textsc{Scattered Cycles} is
\begin{inparaenum}[(a)]
\item NP-hard when restricted to $\ell=2$ and $\Delta=5$; and\label{e:fixedld}
\item W[1]-hard when parameterized by $r$ and $\ell$. \label{e:w1}
\end{inparaenum}
\end{corollary}

\begin{proof}
  \hyperref[e:fixedld]{Item~(\ref*{e:fixedld})}.
  Let $\ell=2$. The reduction of \autoref{l:red} produces in polynomial time an
  instance of the problem \textsc{Scattered Cycles} restricted to
  graphs of maximum degree 5 from an instance of \textsc{Independent
    Set} restricted to graphs of maximum degree 3. Using \hyperref[e:isnp]{item
  (\ref*{e:isnp}) of \autoref*{p:is}} it follows that \textsc{Scattered Cycles} is NP-hard
  even when~$\Delta=5$ and~$\ell=2$.

  \hyperref[e:fixedld]{Item~(\ref*{e:w1})} is a consequence of \autoref{l:red} and of \hyperref[e:isw1]{item
  (\ref*{e:isw1}) of \autoref*{p:is}}.
\end{proof}

\end{document}